\documentclass[submission,copyright,creativecommons]{eptcs}

\providecommand{\event}{SOS 2007} 

\usepackage{mathrsfs}
\usepackage{eurosym} 
\usepackage[normalem]{ulem} 

\usepackage{bussproofs}

\usepackage{amssymb}
\usepackage{stmaryrd}
\usepackage{bussproofs}
\usepackage{textcomp}
\usepackage{amssymb}
\usepackage{amsmath}
\usepackage{stmaryrd}
\usepackage{url}
\usepackage{egpeirce}

\newcommand{\mbf}{\mathbf}

\newcommand{\msf}{\mathsf}
\newcommand{\mrm}{\mathrm}

\newcommand{\imp}{\rightarrow}
\newcommand{\bimp}{\leftrightarrow}

\newcommand{\sub}{\subseteq}

\newcommand{\D}{\Diamond}
\newcommand{\B}{\Box}

\newcommand{\bcut}[1]{\psset{linestyle=dashed,dash=2pt}\cut{#1}}

\newcommand{\scut}[1]{\psset{linestyle=solid}\cut{#1}}

\newcommand{\dbcut}[1]{\psset{linestyle=dashed,dash=2pt}\cut{\cut{{#1}}}}

\newcommand{\pcut}[1]{\psset{linestyle=dashed,dash=2pt}\cut{\scut{#1}}}

\newcommand{\ncut}[1]{\cut{\psset{linestyle=dashed,dash=2pt}\cut{{#1}}}}

\providecommand{\urlalt}[2]{\href{#1}{#2}}
\providecommand{\doi}[1]{doi:\urlalt{http://dx.doi.org/#1}{#1}}

\usepackage{amsthm}
\newtheorem{definition}{Definition}
\newtheorem{remark}{Remark}
\newtheorem{lemma}{Lemma}
\newtheorem{prop}{Proposition}
\newtheorem{theorem}{Theorem}

\title{Graphical Sequent Calculi for Modal Logics}
\author{Minghui Ma
\institute{Institute of Logic and Cognition, Sun Yat-Sen University, Guangzhou, China}
\email{mmh.thu@gmail.com}
\and
Ahti-Veikko Pietarinen
\institute{Chair of Philosophy, Tallinn University of Technology, Tallinn, Estonia}
\email{ahti.pietarinen@gmail.com}
}
\def\titlerunning{Graphical Sequent Calculi for Modal Logics}
\def\authorrunning{M.\ Ma and A.-V.\ Pietarinen}

\begin{document}

\maketitle

\begin{abstract}
The syntax of modal graphs is defined in terms of the continuous cut and broken cut following Charles Peirce's notation in the gamma part of his graphical logic of existential graphs. Graphical calculi for normal modal logics are developed based on a reformulation of the graphical calculus for classical propositional logic. These graphical calculi are of the nature of deep inference. The relationship between graphical calculi and sequent calculi for modal logics is shown by translations between graphs and modal formulas.
 \end{abstract}


\section{Introduction}

Sequent calculi for normal modal logics can be obtained uniformly from a basic calculus, as has been observed in~\cite{Wansing2002}. The search for generalized cut-free sequent calculi for modal logics has produced display calculus (\cite{Belnap1982}), hypersequent calculus (\cite{Avron1996}), labelled sequent calculus (\cite{Sara2005}), hybrid logic calculus (\cite{Seligman1997}), and deep sequent calculus (\cite{Brunnler2003,Brunnler2009,SS2005,Stouppa2007}).
Among these efforts, there are two main approaches. One is the semantic approach; the other largely syntactic. In the semantic approach, labelled calculi exist for a number of complete modal logics. The syntactic approach does not use labels. Each sequent has an obvious corresponding formula. Ordinary sequent calculi and hypersequent calculi for modal logics are syntactic.

Deep inference systems for modal logics, such as deep sequent calculi developed by Br\"unnler \cite{Brunnler2003,Brunnler2009} and Stouppa \cite{Stouppa2007}, are also largely syntactic. There exists also deep inference for hybrid logic (\cite{Strassburger2007}). The syntax of deep sequents is defined by assuming the negation normal form in classical modal logic and nested sequents. The central idea of deep inference is that deep structures are transformed into appropriate shapes at any position in a derivation that allows the transformation. It has turned out that cut-free sequent calculi can be developed systematically and modularly for normal modal logics.

As often is the case, what is syntactic and what semantic may interestingly overlap, as is the case in the hybrid and two-sided approaches. Also in the graphical and diagrammatic systems the distinction between syntax and semantics is not, and was not originally meant to be by Peirce, razor-sharp, which professes to gain some flexibility when dealing with some more complicated and non-standard systems.

The aim of the present paper is to provide a different kind of deep inference system for normal modal logics. The language is given by Peirce's alpha and gamma graphs as presented in his theory of existential graphs (see e.g.\ \cite{Pietarinen2004,Pietarinen2006,Roberts1973,Zeman1964}). Graphs are scribed on the sheet of assertion. Inference rules are formulated as transformation rules from one graph to another graph. In non-modal propositional logic (alpha graphs) and first-order logic (beta graphs), there are basically only two general kinds of transformations: insertions to the graphs and erasures from the graphs. In graphical modal logic, there are two additional kinds of transformations: merges and splits. In a sense also merges and splits are instances of the operations of insertions and erasures. Thus the fundamental proof rules also in the modal extensions of graphical logic can be classified into two general classes. As usual, these operations are allowed only in certain positions in a graph. It is the notion of a position that is made explicit in graphical logic. This makes such graphical calculi the natural home for deep inference.

Peirce's theory of existential graphs was generalized into conceptual graphs by Sowa~\cite{Sowa1984} in 1984. Since then conceptual graphs have been widely used within artificial intelligence and cognitive science. Diagrammatic reasoning and their history and philosophy has been studied for many years (see e.g.~\cite{AB1996,Pietarinen2011,Pietarinen2016}). As far as modal logics are concerned, van den Berg~\cite{Berg93} defines a graphical system for modal logic $\mbf{K}$ which is complete with respect to the Hilbert-style axiomatic system of $\mbf{K}$. Bra\"uner~\cite{Brauner98} defines a Peircean graphical system for the modal logic $\mbf{S5}$, which is also complete with respect to the Hilbert-style axiomatic system of $\mbf{S5}$. This type of graphical system is also extended by Bra{\"u}ner and {\O}hrstr{\o}m~\cite{Brauner99} to modal logics $\mbf{S4}$ and $\mbf{KD45}$. In distinction from the above works, the graphical systems for modal logics presented in this paper are shown to be equivalent to algebraic sequent systems. This means that a range of modal graphical systems can be developed in a systematic and modular fashion.

\section{The syntax of modal graphs}
We fix a denumerable set of simple propositions $\msf{Prop}$ the elements of which are primitive graphs. They occur in a compound graph as basic parts. According to Peirce, the sheet of assertion, or the blank where nothing is scribed on it, is also a primitive graph. It corresponds to tautology $\top$. Henceforth, we denote the blank by {\sf SA} or omit it altogether when no confusion arises. A {\em primitive graph} is a simple proposition or the blank ({\sf SA}).

The modal graphs are defined inductively from primitive graphs using two special notations: the {\em continuous cut} {\cut{\ \ \ }} and the {\em broken cut} {\bcut{\ \ \ }}\ . The continuous cut means negation. The broken cut means logical contingency (non-necessity). The continuous and broken cuts are called {\em primitive cuts} uniformly. There are four combinations of cuts:

(1) Double continuous cut: {\cut{\cut{\ \ \ }}}\ ;

(2) Double broken cut: {\dbcut{\ \ \ }}\ ;

(3) Possibility cut: {\pcut{\ \ \ }}\ ;

(4) Necessity cut: {\ncut{\ \ \ }}\ .
\\
The compound cuts consist of two cuts, one nested within the other, with nothing between them. The two primitive cuts and the four compound cuts stated above are called {\em cuts} uniformly. They are used as single graph operations that form new graphs from the given ones.

\begin{definition}
 The set of all {\em modal graphs} $\mathscr{G}_M$ is defined inductively by:
\[
\mathscr{G}_M\ni G::= p \mid {\sf SA} \mid \cut{\ G\ } \mid \bcut{\ G\ }\mid G_1\,G_2
\]
where $p\in \mathsf{Prop}$. The graphs \cut{\ $G$\ } and \bcut{\ $G$\ } are read as ``the continuous cut of $G$" and ``the broken cut of $G$" respectively. The graph $G_1\,G_2$ is called the {\em juxtaposition} of $G_1$ and $G_2$ on the sheet of assertion.
\end{definition}

Henceforth, when we talk about graphs we mean modal graphs. Given two graphs $G$ and $H$, we define shorthand notations $G\ovee H$, $G \supset H$ and $G\equiv H$ as below:
\[
G\ovee H:= \cut{\ \cut{\ G\ } \ \cut{\ H\ }\ }\ ;~ G\supset H:= \cut{\ G \ \cut{\ H\ }\ }\ ;~  G\equiv H:= \cut{\ G \ \cut{\ H\ }\ }\ \cut{\ H \ \cut{\ G\ }\ }\ .
\]

\begin{definition}
For any graph $G$, the {\em parsing tree} of $G$, denoted by $T(G)$, is defined inductively as follows:
\begin{enumerate}
\item $T({p})$ is a single root node $p$.
\item $T({\sf SA})$ is a single root node ${\sf SA}$.
\item $T(G_1G_2)$ is a root node $G_1G_2$ with children nodes $T(G_1)$ and $T(G_2)$.
\item $T(\cut{\ G\ })$ is a root node $\cut{\ G\ }$ with one child node $T(G)$.
\item $T(\bcut{\ G\ })$ is a root node $\bcut{\ G\ }$ with one child node $T(G)$.
\end{enumerate}
A {\em partial graph} of a graph $G$ is a node in $T(G)$.
\end{definition}

For any graph $G$, the {\em history} of a node $J$ in $T(G)$, denoted by $h(J)$, is the unique path from the root to $J$. The position of the root is always on the sheet of assertion. We say that $J$ is a positive (negative) node of $T(G)$ if there is an even (odd) number of cuts in $h(J)$.

A {\em position} is a point on the area of a graph (but not on the boundary of the cut). Given any graph $G$, a position in $G$ is {\em positive} ({\em negative}) if it is enclosed by an even (odd) number of cuts. Graph are scribed at positions. No two graphs, or their parts, can be scribed at the same position.

A {\em graph context} is a graph $G\{~\}$ with a single {\em slot} $\{~\}$, the empty context, which can be filled by other graphs. The notation $G\{H\}$ stands for the graph obtained from the graph context $G\{~\}$ by filling the slot by $H$. An occurrence of a graph $J$ in a graph $G$ is called {\em positive} ({\em negative}), notation $G\{J^+\}$ ($G\{J^-\}$), if it is a positive (negative) node in $T(G)$.

\section{The graphical calculi $\mbf{K}_g$}

Graphical calculi for modal logics are presented by graphical rules. In general, a graphical rule is of the form
\[
\frac{~G~}{~H~}
\]
where $G$ and $H$ are graphs. The graph $G$ is called the {\em premiss}, and $H$ is called the {\em conclusion}.

On the sheet of assertion, the syntax of graphs becomes diagrammatic. This means that the syntax is two-dimensional, it has no separate notation for parentheses, and that its well-formed graphs are scribed in the ambient space which is continuous, compact, open and non-oriented. The following equalities can be thought of as identifying graphs:
\[
(\mrm{PM})~G\{H_1H_2\}=G\{H_2H_1\};~
(\mrm{AS})~G\{H_1(H_2H_3)\}=G\{(H_1H_2)H_3\}.
\]
The permutation (PM) says that to distinguish positions of $H_1$ and $H_2$ in a partial graph $H_1H_2$ of $G$ has no significance.
The associativity (AS) says that the order of forming the graphs indicated by the parentheses in these rules is likewise immaterial.
After all, these equalities follow from the basic properties of the space and therefore need no separate statement in the system. Likewise, if two graphs, $G$ and $H$, are asserted on the sheet of assertion, the the juxtaposition of them, $G\;H$, is at once also asserted.

The continuous and broken cuts have different meanings in general. However, the continuous cut of {\sf SA} is tantamount to the broken cut of {\sf SA} in the sense that it is impossible to falsify a tautology. Hence we assume the following equality:
\[
(\mrm{Normality})\quad~\cut{\ \ \ }\ =\ \bcut{\ \ \ }
\]
This equality says that contradiction is impossible. Its algebraic meaning is the normality condition in modal algebras (Section 5).

\begin{definition}
The graphical calculus $\mbf{K}_g$ for the minimal normal modal logic $\mbf{K}$ consists of the following axiom and graphical rules:
\begin{enumerate}
\item Axiom: \quad\quad\quad\quad {\sf SA}\quad (The Sheet of Assertion)
\item Alpha rules:
\begin{itemize}
\item {\em Deletion}:
\[
\frac{~G\{H^+\}~}{~G\{{\sf SA}\}~}{(\mrm{DEL})}
\]
Every positive partial graph $H$ in a graph $G$ can be deleted.
\item {\em Insertion}:
\[
\frac{~G\{H^-\}~}{~G\{(JH)^-\}~}{(\mrm{INS})}
\]
Any graph can be inserted into a negative position in a graph $G$.
\item {\em Double cut}:
\[
\AxiomC{$G\{H\}$}
\RightLabel{$(\mrm{DC1})$}
\UnaryInfC{$G\{\cut{\cut{\ H\ }}\}$}
\DisplayProof
\quad
\AxiomC{$G\{\cut{\cut{\ H\ }}\}$}
\RightLabel{$(\mrm{DC2})$}
\UnaryInfC{$G\{H\}$}
\DisplayProof
\]
Any partial graph $H$ of a graph $G$ can be replaced by the double cut of $H$, and vice versa.
\item {\em Iteration/deiteration}:
\[
\frac{K\{GH\{J\}\}}{K\{GH\{GJ\}\}}{(\mrm{IT})}
\quad
\frac{K\{GH\{GJ\}\}}{K\{GH\{J\}\}}{(\mrm{DEIT})}
\]
where $H\{~\}$ is a broken-cut-free graph context, namely, no broken cut occurs in $H\{~\}$. In a graph $K\{GH\{J\}\}$, the partial graph $G$ can be iterated or deiterated at any position in $H$.
\end{itemize}
\item Modal rules:
\[
\AxiomC{$J\{\ncut{\ GH\ }\}$}
\RightLabel{$(\mrm{K1})$}
\UnaryInfC{$J\{\ncut{\ G\ }~\ncut{\ H\ }\}$}
\DisplayProof
\quad
\AxiomC{$J\{\ncut{\ G\ }~\ncut{\ H\ }\}$}
\RightLabel{$(\mrm{K2})$}
\UnaryInfC{$J\{\ncut{\ GH\ }\}$}
\DisplayProof
\quad
\AxiomC{$J\{(G\supset H)^+ \}$}
\RightLabel{$(\mrm{DMN})$}
\UnaryInfC{$J\{(\bcut{\ H\ }\supset \bcut{\ G\ })^+\}$}
\DisplayProof
\]
$(\mrm{K1})$ and $(\mrm{K2})$ mean that the necessity cut distributes over juxtaposition. We call the rule $(\mrm{K1})$ {\em splitting} and $(\mrm{K2})$ {\em merging}. $(\mrm{DMN})$ is the rule of {\em downward monotonicity}.
\end{enumerate}
\end{definition}

A {\em proof} of a graph $G$ in $\mbf{K}_g$ is a finite sequence of graphs $G_0,\ldots,G_n$ such that $G_n=G$, and each $G_i$ is either {\sf SA} or derived from previous graphs by a rule in $\mbf{K}_g$.
A graph $G$ is {\em provable} in $\mbf{K}_g$, notation $\vdash_{\mbf{K}_g} G$, if it has a proof in $\mbf{K}_g$.
A graphical derivation of $H$ from $G$ is {\em admissible} in $\mbf{K}_g$, if $\vdash_{\mbf{K}_g}G$ implies $\vdash_{\mbf{K}_g} H$.

\begin{remark}
The restriction on the context $H\{~\}$ in (IT) and (DEIT) rules is significant. Iteration/de-iteration in a modal context may lead to invalid inferences. For example, consider the following two inferences where the rules (IT) and (DEIT) are applied into the broken cut:
\[
\frac{p~\pcut{\ q\ }\supset p~\pcut{\ q\ }}{p~\pcut{\ q\ }\supset p~\pcut{\ pq\ }}{(\mrm{I})}
\quad
\frac{p~\bcut{\ q\ }\supset p~\bcut{\ q\ }}{p~\bcut{\ p~q\ }\supset p~\bcut{\ q\ }}{(\mrm{II})}
\]
The premisses of (I) and (II) are valid, but their conclusions are not valid in the algebraic semantics for $\mbf{K}_g$ (Section 5). (I) is a counterexample to the validity of iteration into broken cut, and (II) is a counterexample to the validity of deiteration from a broken cut.
\end{remark}

\begin{lemma}\label{lemma:top}
The graphs $G\supset {\sf SA}$ and $G\supset G$ are derivable in $\mbf{K}_g$.
\end{lemma}
\begin{proof}
We have the following proofs:
\[
\AxiomC{{\sf SA}}
\RightLabel{$(\mrm{DC1})$}
\UnaryInfC{$\cut{\ \ \cut{\ \ }\ }$}
\RightLabel{$(\mrm{INS})$}
\UnaryInfC{$\cut{\ G\ \cut{\ \ }\ }$}
\DisplayProof
\quad
\AxiomC{{\sf SA}}
\RightLabel{$(\mrm{DC1})$}
\UnaryInfC{$\cut{\ \ \cut{\ \ }\ }$}
\RightLabel{$(\mrm{INS})$}
\UnaryInfC{$\cut{\ G\ \cut{\ \ \ }\ }$}
\RightLabel{$(\mrm{IT})$}
\UnaryInfC{$\cut{\ G\ \cut{\ G\ }\ }$}
\DisplayProof
\]
This completes the proof.
\end{proof}

\begin{prop}
The following rules are admissible in $\mbf{K}_g$:
\begin{enumerate}
\item De Morgan rules:
\[
\AxiomC{$\cut{\ GH \ }$}
\RightLabel{$(\mrm{DM1})$}
\UnaryInfC{$\cut{\ G \ }\ovee \cut{\ H\ }$}
\DisplayProof
\quad
\AxiomC{$\cut{\ G \ }\ovee \cut{\ H\ }$}
\RightLabel{$(\mrm{DM2})$}
\UnaryInfC{$\cut{\ GH \ }$}
\DisplayProof
\]
\item Contraposition and transitivity rules:
\[
\frac{G\supset H}{\cut{\ H \ }\ \supset\ \cut{\ G\ \ }}{(\mrm{CP})}
\quad
\frac{G\supset H\quad H\supset J}{G\supset J}{(\mrm{TR})}
\]
\item Prefixing and Modus Ponens:
\[
\frac{G}{H\supset G}{(\mrm{PF})}
\quad
\frac{G\quad G\supset H}{H}{(\mrm{MP})}
\]
\item Lattice rules:
\[
\frac{G_i\supset H}{G_1G_2\supset H}{\mrm{(\&L)}}
\quad
\frac{G\supset H\quad G\supset J}{G\supset HJ}{(\mrm{\&R})}
\quad
\frac{G\supset J\quad H\supset J}{G\ovee H\supset J}{(\ovee\mrm{L})}
\]
\[
\frac{G\supset H_i}{G\supset H_1\ovee H_2}{(\ovee\mrm{L})}
\quad
\frac{G\supset H\ovee J}{\cut{\ H\ }\ G\supset J}{(\mrm{NL})}
\quad
\frac{GH\supset J}{H\supset \cut{\ G\ }\ovee J}{(\mrm{NR})}
\]
\item Residuation rules:
\[
\frac{GH\supset J}{G\supset (H\supset J)}{(\mrm{RG1})}
\quad
\frac{G\supset (H\supset J)}{GH\supset J}{(\mrm{RG2})}
\]
\item Distributivity:
\[
\AxiomC{$G\ \cut{\ \cut{\ H\ }\ \cut{\ J\ } \ }$}
\RightLabel{$(\mrm{D1})$}
\UnaryInfC{$\cut{\ \cut{\ G\ H\ }\ \cut{\ G\ J \ } \ }$}
\DisplayProof
\quad
\AxiomC{$\cut{\ \cut{\ G\ H\ }\ \cut{\ G\ J \ } \ }$}
\RightLabel{$(\mrm{D2})$}
\UnaryInfC{$G\ \cut{\ \cut{\ H\ }\ \cut{\ J\ } \ }$}
\DisplayProof
\]
\item Upward monotonicity:
\[
\AxiomC{$G\supset H$}
\RightLabel{$(\mrm{UMN})$}
\UnaryInfC{$\ncut{\ G \ }\supset \ncut{\ H\ }$}
\DisplayProof
\quad
\AxiomC{$G\supset H$}
\RightLabel{$(\mrm{UMP})$}
\UnaryInfC{$\pcut{\ G \ }\supset \pcut{\ H\ }$}
\DisplayProof
\quad
\AxiomC{$G\supset H$}
\RightLabel{$(\mrm{UMDB})$}
\UnaryInfC{$\dbcut{\ G \ }\supset \dbcut{\ H\ }$}
\DisplayProof
\]
\item Replacement of equivalents:
\[
\frac{G\equiv H}{J\{G\}\equiv J\{H\}}{(\mrm{RE})}
\]
\item Necessitation rule:
\[
\frac{G}{\ncut{\ G\ }}{(\mrm{Nec})}
\]
\end{enumerate}
\end{prop}
\begin{proof}
For (DM1) and (DM2), we have the following simple proofs that only use the double-cut rules:
\[
\AxiomC{$\cut{\ GH \ }$}
\RightLabel{$(\mrm{DC1})$}
\UnaryInfC{$\cut{\ \cut{\ \cut{\ G\ }\ }\ H \ }$}
\RightLabel{$(\mrm{DC1})$}
\UnaryInfC{$\cut{\ \cut{\ \cut{\ G\ }\ }\ \cut{\ \cut{\ H\ } \ } \ }$}
\DisplayProof
\quad
\AxiomC{$\cut{\ \cut{\ \cut{\ G\ }\ }\ \cut{\ \cut{\ H\ } \ } \ }$}
\RightLabel{$(\mrm{DC2})$}
\UnaryInfC{$\cut{\ \cut{\ \cut{\ G\ }\ }\ H \ }$}
\RightLabel{$(\mrm{DC2})$}
\UnaryInfC{$\cut{\ GH \ }$}
\DisplayProof
\]
(TR) is shown as follows:
\[
\AxiomC{$\cut{\ G\  \cut{\ J\ }\ }$\quad$\cut{\ J\  \cut{\ H\ }\ }$}
\RightLabel{\small $(\mrm{IT})$}
\UnaryInfC{$\cut{\ G\  \cut{\ J\ \cut{\ J\  \cut{\ H\ }\ }\ }\ }$}
\RightLabel{\small $(\mrm{DEIT})$}
\UnaryInfC{$\cut{\ G\  \cut{\ J\ \cut{\ \  \cut{\ H\ }\ }\ }\ }$}
\RightLabel{\small $(\mrm{DEL})$}
\UnaryInfC{$\cut{\ G\  \cut{\ \ \cut{\ \  \cut{\ H\ }\ }\ }\ }$}
\RightLabel{\small $(\mrm{DC2})$}
\UnaryInfC{$\cut{\ G\  \cut{\ H\ }\ }$}
\DisplayProof
\]
For (D1) and (D2), we have the following proofs (\cite{MaPietarinen2016}):
\[
\AxiomC{$G\ \cut{\ \cut{\ H\ }\ \cut{\ J \ } \ }$}
\RightLabel{$(\mrm{IT})$}
\UnaryInfC{$G\ \cut{\ \cut{\ GH\ }\ \cut{\ J \ } \ }$}
\RightLabel{$(\mrm{IT})$}
\UnaryInfC{$G\ \cut{\ \cut{\ GH\ }\ \cut{\ GJ \ } \ }$}
\RightLabel{$(\mrm{DEL})$}
\UnaryInfC{$\cut{\ \cut{\ GH\ }\ \cut{\ GJ \ } \ }$}
\DisplayProof
\quad
\AxiomC{$\cut{\ \cut{\ GH\ }\ \cut{\ GJ\ } \ }$}
\RightLabel{$(\mrm{IT})$}
\UnaryInfC{$\cut{\ \cut{\ GH\ }\ \cut{\ GJ\ } \ }$\quad$\cut{\ \cut{\ GH\ }\ \cut{\ GJ\ } \ }$}
\RightLabel{$(4\mrm{~times~DEL})$}
\UnaryInfC{$\cut{\ \cut{\ G\ }\ \cut{\ G\ } \ }$\quad$\cut{\ \cut{\ H\ }\ \cut{\ J\ } \ }$}
\RightLabel{$(\mrm{DEIT})$}
\UnaryInfC{$\cut{\ \cut{\ G\ } \ }$\quad$\cut{\ \cut{\ H\ }\ \cut{\ J\ } \ }$}
\RightLabel{$(\mrm{DC2})$}
\UnaryInfC{$G\ \cut{\ \cut{\ H\ }\ \cut{\ J\ } \ }$}
\DisplayProof
\]
The rule (RE) is shown by induction on the construction of $J\{~\}$ as follows. Assume $G\equiv H$. If $J\{~\} = \{~\}$, the conclusion is the same as the premiss. Suppose $J\{~\} = \cut{\ J'\{~\}\ }$. By induction hypothesis, we have $J'\{G\}\equiv J'\{H\}$. Then it is easy to show $\cut{\ J'\{G\}\ }\equiv\cut{\ J'\{H\}\ }$. Assume $J\{~\} = J_1J_2\{~\}$. By induction hypothesis, we have $J_2\{G\}\equiv J_2\{H\}$. Then it is easy to show $J_1J_2\{G\}\equiv J_1J_2\{H\}$.

The rule (UMN) is obtained from (DMN) by the rule of contraposition (CP). (Nec) is shown by (PF), (UMN) and (TR).
The other rules are easily shown.
\end{proof}

\begin{theorem}[Cut-elimination]
The following cut-elimination rule
\[
\AxiomC{$J\{\cut{\ G\ \cut{\ G\ }\ }\}$}
\RightLabel{$\mrm{(Cut}$-$\mrm{E)}$}
\UnaryInfC{$J\{ {\sf SA} \}$}
\DisplayProof
\]
is admissible in $\mbf{K}_g$.
\end{theorem}
\begin{proof}
Clearly $\cut{\ G\ \cut{\ G\ }\ }\ \equiv\ {\sf SA}$ is provable in $\msf{K}_g$. By (RE), we have $J\{\cut{\ G\ \cut{\ G\ }\ }\}\ \equiv\ J\{{\sf SA}\}$. Assume $\vdash_{\mbf{K}_g}J\{\cut{\ G\ \cut{\ G\ }\ }\}$. By (TR), we have $\vdash_{\mbf{K}_g} J\{ {\sf SA} \}$.
\end{proof}

\section{Extensions}

Extensions of $\mbf{K}_g$ can be obtained by adding some characteristic rules. The formulation of these characteristic rules will make use of the cuts, including the six cuts (two primitive and four combined ones) we introduced in Section~2.
We say that the occurrence of a cut in a graph is {\em positive} ({\em negative}) if it is enclosed evenly (oddly) by primitive cuts (continuous or broken cuts).

A {\em normal modal graphical calculus} is an extension of $\mbf{K}_g$ with a set of graphical rules. Given a set of rules $\Sigma = \{R_i\mid i\in I\}$, the notation $\mbf{K}\Sigma$ denotes the calculus generated by rules in $\Sigma$. Let us have the following rules of transformation as the basic rules for various systems of graphical modal logic:
\begin{enumerate}
\item[$(D)$] Any positive necessity cut can be transformed into a possibility cut. Any negative possibility cut can be transformed into a necessity cut.
\[
\frac{J\{\ncut{\ G\ }^+\}}{J\{\pcut{\ G\ }^+\}}{(D^+)}
\quad
\frac{J\{\pcut{\ G\ }^-\}}{J\{{\ncut{\ G\ }^-}\}}{(D^-)}
\]
\item[$(T)$] Any positive continuous cut can be transformed into a broken cut.
Any negative broken cut can be transformed into a continuous cut.
\[
\frac{J\{\cut{\ G\ }^+\}}{J\{\bcut{\ G\ }^+\}}{(T^+)}
\quad
\frac{J\{\bcut{\ G\ }^-\}}{J\{\scut{\ G\ }^-\}}{(T^-)}
\]
\item[$(4)$] Any positive necessity cut can be doubled. Any negative possibility cut can be doubled.
\[
\frac{J\{\ncut{\ G\ }^+\}}{J\{\ncut{{\scut{\bcut{\ G\ }}}}^+\}}{(4^+)}
\quad
\frac{J\{\pcut{\ G\ }^-\}}{J\{{\pcut{\pcut{\ G\ }}}^-\}}{(4^-)}
\]
\item[$(B)$] Any positive double broken cut can be deleted. Any double broken cut can be inserted into a negative position.
\[
\frac{J\{\dbcut{\ G\ }^+\}}{J\{{G}^+\}}{(B^+)}
\quad
\frac{J\{G^-\}}{J\{{\dbcut{\ G\ }}^-\}}{(B^-)}
\]
\item[$(5)$] Any positive double broken cut can be transformed into a necessity cut. Any negative possibility cut can be transformed into a double broken cut.
\[
\frac{J\{\dbcut{\ G\ }^+\}}{J\{{\ncut{\ G\ }^+}\}}{(5^+)}
\quad
\frac{J\{\ncut{\ G\ }^-\}}{J\{{\dbcut{\ G\ }^-}\}}{(5^-)}
\]
\end{enumerate}

\begin{definition}\label{def:systems}
Let $(X) = \{(X^+), (X^-)\}$ for $X\in \{D, T, 4, B, 5\}$. We define the following graphical calculi:
\begin{align*}
\msf{KD}_g &= \msf{K}_g(D)
&
\msf{KB}_g &= \msf{K}_g(B)
&
\msf{K4}_g &= \msf{K}_g(4)
\\
\msf{K5}_g &= \msf{K}_g(5)
&
\msf{KT}_g &= \msf{K}_g(T)
&
\msf{KDB}_g &= \msf{KD}_g(B)
\\
\msf{KB4}_g &= \msf{KB}_g(4)
&
\msf{KD4}_g &= \msf{KD}_g(4)
&
\msf{KD5}_g &= \msf{KD}_g(5)
\\
\msf{KB5}_g &= \msf{KB}_g(5)
&
\msf{K45}_g &= \msf{K4}_g(5)
&
\msf{KTB}_g &= \msf{KT}_g(B)
\\
\msf{S4}_g &= \msf{KT}_g(4)
&
\msf{S5}_g &= \msf{KT}_g(5)
\end{align*}
\end{definition}
Let $S$ be any one of the systems in Definition~\ref{def:systems}. Let $S^+$
and $S^-$ be the systems obtained from $S$ by dropping the negative and positive rules respectively.

\begin{theorem}
$S^+ = S = S^-$.
\end{theorem}
\begin{proof}
Consider $\mbf{KT}^+ = \mbf{K}_g{(T^+)}$. It suffices to show that $(T^-)$ is provable in $\mbf{KT}^+$. Assume that $J\{\bcut{\ G\ }^-\}$ is provable in  $\mbf{KT}^+$. There are two cases:

Case 1.\ $J\{\bcut{\ G\ }^-\} = J'~\bcut{\ H\ \bcut{\ G\ }\ }$. First, it is easy to prove $\scut{\ G\ }\supset \bcut{\ G\ }$ in $\mbf{KT}^+$.
Then we have the following proof:
\[
\AxiomC{$\scut{\ G\ }\supset \scut{\ G\ }$}
\RightLabel{(\&L)}
\UnaryInfC{$H~\scut{\ G\ }\supset \scut{\ G\ }$}
\AxiomC{$\scut{\ G\ }\supset \bcut{\ G\ }$}
\RightLabel{(TR)}
\BinaryInfC{$H~\scut{\ G\ }\supset \bcut{\ G\ }$}
\AxiomC{$H\supset H$}
\RightLabel{(\&L)}
\UnaryInfC{$H~\scut{\ G\ }\supset H$}
\RightLabel{(\&R)}
\BinaryInfC{$H~\scut{\ G\ }\supset H~\bcut{\ G\ }$}
\RightLabel{(DMN)}
\UnaryInfC{$\bcut{\ H~\bcut{\ G\ }\ }\supset \bcut{\ H~\scut{\ G\ }\ }$}
\RightLabel{(Alpha rules)}
\UnaryInfC{$J'~\bcut{\ H~\bcut{\ G\ }\ }\supset J'~\bcut{\ H~\scut{\ G\ }\ }$}
\DisplayProof
\]

Case 2.\ $J\{\scut{\ G\ }^-\} = J'~\scut{\ H\ \bcut{\ G\ }\ }$. We have the following proof:
\[
\AxiomC{$\scut{\ G\ }\supset \bcut{\ G\ }$}
\RightLabel{(Alpha rules)}
\UnaryInfC{$H~\scut{\ G\ }\supset H~\bcut{\ G\ }$}
\RightLabel{(CP)}
\UnaryInfC{$\scut{\ H~\bcut{\ G\ }\ }\supset \scut{\ H~\scut{\ G\ }\ }$}
\RightLabel{(Alpha rules)}
\UnaryInfC{$J'~\scut{\ H~\bcut{\ G\ }\ }\supset J'~\scut{\ H~\scut{\ G\ }\ }$}
\DisplayProof
\]
Hence $(T^-)$ is provable in $\mbf{KT}^+$.
The remaining cases of $S$ are shown similarly.
\end{proof}

\section{Graphical and sequent calculi}

The set of all modal formulas $\mathscr{L}_M$ is defined by the following inductive rule:
\[
\mathscr{L}_M\ni\alpha::= p\mid \top\mid\neg \alpha\mid (\alpha\wedge\alpha) \mid \B\alpha,
\]
where $p\in \msf{Prop}$. Other propositional connectives $\bot,\vee, \imp$ and $\bimp$ are defined as usual.
The dual operator of $\D$ is defined as $\D\alpha :=\neg\B\neg\alpha$. A {\em basic sequent} is an expression of the form $\alpha\vdash\beta$.

\begin{definition}
The basic sequent calculus $\mbf{SK}$ consists of the following axioms and rules:
\begin{itemize}
\item[$(1)$] Axioms:
\[
(\mrm{Id})~\alpha\vdash \alpha,\quad (\top)~\alpha\vdash \top,
\quad
(\mrm{D})~\alpha\wedge(\beta\vee\gamma)\vdash(\alpha\wedge\beta)\vee(\alpha\wedge\gamma),
\]
\[
(Em)~\top\vdash \alpha\vee \neg \alpha,
\quad
(\mrm{Gen})~\top\vdash\Box\top,
\quad
(\B\wedge)~\B\alpha\wedge\B\beta\vdash\B(\alpha\wedge\beta).
\]
\item[$(2)$] Rules for propositional connectives:
\[
\AxiomC{$\neg\alpha\vdash \beta$}
\RightLabel{$(\neg\mrm{L})$}
\UnaryInfC{$\neg\beta\vdash \alpha$}
\DisplayProof
\quad
\AxiomC{$\alpha\vdash \neg\beta$}
\RightLabel{$(\neg\mrm{R})$}
\UnaryInfC{$\beta\vdash \neg\alpha$}
\DisplayProof
\quad
\frac{\alpha\vdash \beta\quad \beta\vdash \gamma}{\alpha\vdash \gamma}{(\mrm{Tr})}
\]
\[
\AxiomC{$\alpha_i\vdash \beta$}
\RightLabel{$(\wedge\mrm{L})(i=1,2)$}
\UnaryInfC{$\alpha_1\wedge \alpha_2\vdash \beta$}
\DisplayProof
\quad
\AxiomC{$\beta\vdash \alpha_1$\quad $\beta\vdash \alpha_2$}
\RightLabel{$(\wedge\mrm{R})$}
\UnaryInfC{$\beta\vdash \alpha_1\wedge \alpha_2$}
\DisplayProof
\]
\item[$(3)$] Modal rule:
\[
\frac{\alpha\vdash\beta}{\B \alpha \vdash\B \beta}{(\B)}
\]
\end{itemize}
\end{definition}

By the standard Lindenbaum--Tarski construction, one can easily obtain the following completeness result:

\begin{theorem}\label{theorem:compl}
A sequent is derivable in $\mbf{SK}$ iff it is valid in all modal algebras.
\end{theorem}

We shall present the translations between the modal language $\mathscr{L}_M$ and the graphical language $\mathscr{G}_M$, and then prove the connections between the graphical calculus $\mbf{K}_g$ and the sequent calculus $\mbf{SK}$.

\begin{definition}
The translation $\pi\colon \mathscr{G}_M\imp \mathscr{L}_M$ is defined inductively by
\[
\pi({p}) = p;~ \pi({\sf SA}) = \top;~
\pi(\cut{\ G\ }) = \neg \pi(G);~
\]
\[
 \pi(\bcut{\ G\ }) = \D\neg \pi(G);~
\pi(G_1\;G_2) = \pi(G_1)\wedge\pi(G_2).
\]
The translation $\sigma\colon \mathscr{L}_M\imp \mathscr{G}_M$ is defined inductively by
\[
\sigma({p}) = p;~\sigma(\top) = {\sf SA};~\sigma(\neg\varphi) = \cut{\ \sigma(\varphi)\ }\ ;~
\]
\[\sigma(\Box\varphi) = \ncut{\ \sigma(\varphi)\ }\ ;~
\sigma(\varphi_1\wedge\varphi_2) = \sigma(\varphi_1)\;\sigma(\varphi_2).
\]
\end{definition}

The two translations $\pi$ and $\sigma$ are related to each other. The relationship can be presented by the following result:

\begin{prop}\label{prop:trans}
There are functions $\delta\colon \mathscr{L}_M\imp \mathscr{L}_M$ and $\rho\colon \mathscr{G}_M\imp \mathscr{G}_M$ such that the following diagrams commute:
%
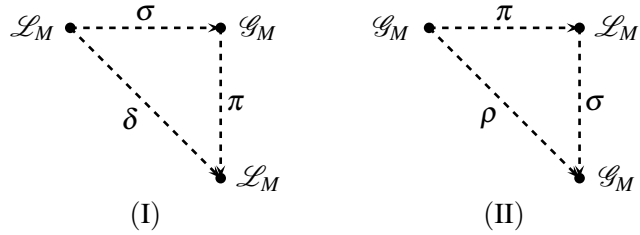
\[
\psset{linewidth=1pt,linestyle=dashed,dash=2.5pt}{
\begin{pspicture}(-1,-0.5)(3,2.5)
\psdot(0,2)\psdot(2,2)\psdot(2,0)
\rput(-0.5,2){$\mathscr{L}_M$}\rput(2.5,2){$\mathscr{G}_M$}\rput(2.5,0){$\mathscr{L}_M$}
\psline{->}(0,2)(2,2)\rput(1,2.2){$\sigma$}
\psline{->}(0,2)(2,0)\rput(0.8,0.8){$\delta$}
\psline{->}(2,2)(2,0)\rput(2.2,1){$\pi$}
\rput(1,-0.5){$(\mrm{I})$}
\end{pspicture}
\quad\quad
\begin{pspicture}(-1,-0.5)(3,2.5)
\psdot(0,2)\psdot(2,2)\psdot(2,0)
\rput(-0.5,2){$\mathscr{G}_M$}\rput(2.5,2){$\mathscr{L}_M$}\rput(2.5,0){$\mathscr{G}_M$}
\psline{->}(0,2)(2,2)\rput(1,2.2){$\pi$}
\psline{->}(0,2)(2,0)\rput(0.8,0.8){$\rho$}
\psline{->}(2,2)(2,0)\rput(2.2,1){$\sigma$}
\rput(1,-0.5){$(\mrm{II})$}
\end{pspicture}
}
\]

i.e., $\pi\circ\sigma = \delta$ and $\sigma\circ\pi = \rho$.
\end{prop}
\begin{proof}
As we are using them later on, let us first define the two (redundant) functions $\delta$ and $\rho$ as follows. Define the function $\delta$ inductively by:
$\delta({p}) = p$,  $\delta(\top) = \top$, $\delta(\varphi_1\wedge\varphi_2) = \delta(\varphi_1)\wedge\delta(\varphi_2)$, and $\delta(\Box\phi) = \neg\D\neg\varphi$.
By induction on the construction of a modal formula $\varphi$ one can easily show $\sigma(\pi(\varphi))=\delta(\varphi)$. Hence $(\mrm{I})$ commutes.
Define the function $\rho$ inductively as follows:
\begin{align*}
\rho({p}) &= p, & \rho({\sf SA}) &= {\sf SA}, &\rho(G_1\;G_2) &= G_1\;G_2, \\\rho(\cut{\ G\ }) &= \cut{\ G\ }, & \rho({\bcut{\ G\ }}) &= \cut{\ \ncut{\ {\scut{\scut{\ G\ }}}\ }\ }\ . & &
\end{align*}
By induction on the construction of a graph $G$ one can easily show that $\sigma(\pi(G))=\rho(G)$. Hence $(\mrm{II})$ commutes.
\end{proof}

A {\em formula context} is a formula structure $\alpha\{~\}$ with a single slot $\{~\}$ which can be filled with a formula. Let $\alpha\{\beta\}$ be the formula obtained from $\alpha\{~\}$ by filling the slot by $\beta$. The notation $\alpha\{\beta^+\}$ stands for that $\beta$ is positive in $\alpha$, i.e., $\beta$ is in the scope of an even number of negation symbols. Similarly we use the notation $\alpha\{\beta^-\}$.

\begin{lemma}\label{lemma:holeSK}
The following hold in $\mbf{SK}$:
\begin{itemize}
\item[$(1)$] if $\alpha\{\beta^+\}$ and $ \beta\vdash_{\mbf{SK}}\gamma$, then $\alpha\{\beta\} \vdash_{\mbf{SK}} \alpha\{\gamma\}$.
\item[$(2)$] if $\alpha\{\beta^-\}$ and $\beta\vdash_{\mbf{SK}}\gamma$, then $\alpha\{\gamma\} \vdash_{\mbf{SK}} \alpha\{\beta\}$.
\item[$(3)$] if $\beta\vdash_{\mbf{SK}}\gamma$ and $\gamma\vdash_{\mbf{SK}}\beta$, then $\alpha\{\beta\} \vdash_{\mbf{SK}} \alpha\{\gamma\}$ and $\alpha\{\gamma\} \vdash_{\mbf{SK}} \alpha\{\beta\}$.
\end{itemize}
\end{lemma}
\begin{proof}
By induction on the construction of $\alpha\{~\}$.
We sketch the proof of (1) and (2) by simultaneous induction. The case $\alpha\{~\} = \{~\}$ is obvious. Suppose $\alpha\{\beta\} := \neg\alpha'\{\beta\}$ and $\beta\vdash_\mbf{SK}\gamma$. There are two cases:

Case 1.\ $\neg\alpha'\{\beta^+\}$. Then $\alpha'\{\beta^-\}$. By induction hypothesis, we have $\alpha'\{\gamma\}\vdash_\mbf{SK}\alpha'\{\beta\}$. Then $\neg\alpha'\{\beta\}\vdash_\mbf{SK}\neg\alpha'\{\gamma\}$.

Case 2.\ $\neg\alpha'\{\beta^-\}$. Then $\alpha'\{\beta^+\}$. By induction hypothesis, we have $\alpha'\{\beta\}\vdash_\mbf{SK}\alpha'\{\gamma\}$. Then $\neg\alpha'\{\gamma\}\vdash_\mbf{SK}\neg\alpha'\{\beta\}$.

The case $\alpha\{~\} = \alpha_1\{~\}\wedge\alpha_2$ or
$\alpha\{~\} = \alpha_1\wedge\alpha_2\{~\}$ is obvious. Suppose $\alpha\{~\} = \Box\alpha'\{~\}$ and $\beta\vdash_\mbf{SK}\gamma$.
Assume $\Box\alpha'\{\beta^+\}$. Then by induction hypothesis we have $\alpha'\{\beta\}\vdash_\mbf{SK}\alpha'\{\gamma\}$.
Then by $(\Box)$ we have $\Box\alpha'\{\beta\}\vdash_\mbf{SK}\Box\alpha'\{\gamma\}$. The case for $\Box\alpha'\{\beta^-\}$ is similar.
\end{proof}

\begin{lemma}\label{lemma:g-sk}
For any graph $G$,
if $\vdash_{\mbf{K}_g} G$, then $\top\vdash_\mbf{SK} \pi(G)$.
\end{lemma}
\begin{proof}
Assume $\vdash_{\mbf{K}_g} G$. Let $G_0,\ldots, G_n = G$ be a proof of $G$. We show $\top\vdash_\mbf{SK} \pi(G_i)$ by induction on $i\leq n$. If $G_i$ is {\sf SA}, clearly we have $\top\vdash_\mbf{SK} \pi({G_i})$. Assume that $G_i$ is obtained from $G'$ by a rule $({R})$. If $({R})$ is an alpha rule, it is easy to get the conclusion by induction hypothesis and Lemma~\ref{lemma:holeSK}. Suppose that $({R})$ is a modal rule.

(1).\ $({R})=(\mrm{K1})$ or $(\mrm{K2})$. Let $G_i = J\{\ncut{\ H\ }~\ncut{\ K\ }\}$ and $G'=J\{\ncut{\ H K\ }\}$. By induction hypothesis, we have $\top\vdash_\mbf{SK} \pi(J)\{\pi(\ncut{\ H K\ })\}$, i.e., $\top\vdash_\mbf{SK} \pi(J)\{\neg\D\neg(\pi(H)\wedge\pi(K))\}$.
Clearly $\neg\D\neg(\pi(H)\wedge\pi(K))\vdash_\mbf{SK} \neg\D\neg\pi(H)\wedge \neg\D\neg \pi(K)$
and $\neg\D\neg\pi(H)\wedge \neg\D\neg \pi(K)\vdash_\mbf{SK} \neg\D\neg(\pi(H)\wedge\pi(K))$. By Lemma~\ref{lemma:holeSK} (3), we get $\top\vdash_\mbf{SK} \pi(G_i)$. The case for $(\mrm{K2})$ is similar.

(2).\ $({R})=(\mrm{DMN})$. Let $G_i=J\{(\bcut{\ K\ } \supset \bcut{\ H\ })^+\}$ and $G' = J\{(H\supset K)^+\}$. By induction hypothesis, we have $\top\vdash_\mbf{SK} \pi(J\{(H\supset K)^+\})$, i.e., $\top\vdash_\mbf{SK} \pi(J)\{\neg (\pi(H)\wedge \pi(K))\}$. Clearly, $\neg (\pi(H)\wedge \neg \pi(K)) \vdash_\mbf{SK} \neg (\D\neg \pi(K)\wedge \D\neg \neg \pi(H))$. By Lemma~\ref{lemma:holeSK} (1), we get $\top\vdash_\mbf{SK} \pi(G_i)$.
\end{proof}

\begin{lemma}\label{lemma:sk-g}
For any formula $\alpha$,
if $\top\vdash_\mbf{SK}\alpha$, then $\vdash_{\mbf{K}_g}\sigma(\alpha)$.
\end{lemma}
\begin{proof}
By induction on the derivation of $\top\vdash\alpha$ in $\mbf{SK}$. The proof is omitted.
\end{proof}

\begin{lemma}\label{lemma:rho}
For any graph $G$,
$\vdash_{\mbf{K}_g} G$ iff $\vdash_{\mbf{K}_g} \rho(G)$.
\end{lemma}
\begin{proof}
By induction on the proof of $G$ in $\mbf{K}_g$. The proof is omitted.
\end{proof}

\begin{theorem}\label{thm:embedding}
For any graph $G$,
$\vdash_{\mbf{K}_g} G$ iff $\top\vdash_\mbf{SK} \pi(G)$.
\end{theorem}
\begin{proof}
The `only if' part is obtained by Lemma~\ref{lemma:g-sk}. Assume $\top\vdash_\mbf{SK} \pi(G)$. By Lemma~\ref{lemma:sk-g}, we have $\vdash_{\mbf{K}_g} \sigma\circ\pi(G)$. By Proposition~\ref{prop:trans}, $\vdash_{\mbf{K}_g} \rho(G)$. By Lemma~\ref{lemma:rho}, $\vdash_{\mbf{K}_g} G$.
\end{proof}

\begin{definition}
A {\em modal algebra} is an algebra $\mbf{A}=(A,\wedge, \neg, \B, 1)$ where $(A, \wedge, \neg, 1)$ is a Boolean algebra, and $\B$ is a unary operator on $A$ satisfying the conditions:
\begin{enumerate}
\item Additivity: for all $a,b\in A$, $\B(a\wedge b)=\B a \wedge \B b$;
\item Normality: $\B 1 = 1$.
\end{enumerate}
\end{definition}

Any formula $\alpha$ is interpreted as a function $\alpha^\mbf{A}$ in a modal algebra $\mbf{A}$. A sequent $\alpha\vdash\beta$ is {\em valid} in $\mbf{A}$ if $\alpha^\mbf{A}\leq \beta^\mbf{A}$ whatever elements of $A$ are assigned to variables in $\alpha$ or $\beta$.
By the standard Lindenbaum--Tarski construction, one can show the completeness of $\mbf{SK}$ with respect to the class of all modal algebras, i.e., $\alpha\vdash_\mbf{SK}\beta$ if and only if $\alpha\vdash\beta$ is valid in all modal algebras (Theorem~\ref{theorem:compl}).

A graph $G$ is interpreted as the function $G^\mbf{A} = \pi(G)^\mbf{A}$. A graph $G$ is {\em valid} in a modal algebra $\mbf{A}$ if $\top\vdash \pi(G)$ is valid in $\mbf{A}$. Then one can obtain the following completeness result:

\begin{theorem}
A graph $G$ is provable in $\mbf{K}_g$ iff it is valid in all modal algebras.
\end{theorem}
\begin{proof}
The soundness is shown by induction on the proof of $G$. For completeness, assume $\not\vdash_{\mbf{K}_g} G$. By Theorem~\ref{thm:embedding}, we have $\top\not\vdash_{\mbf{SK}} \pi(G)$. By the completeness of $\mbf{SK}$, there is a modal algebra $\mbf{A}$ with $1\not\leq \pi(G)^\mbf{A}$. Then $G$ is not valid in $\mbf{A}$.
\end{proof}

For any set of modal formulas $\Sigma$, let $\Sigma^\leq = \{\top\vdash\alpha\mid \alpha\in \Sigma\}$. Then we have the basic sequent calculus $\mbf{SK}\Sigma^\vdash$ which is obtained from $\mbf{SK}$ by adding all sequents in $\Sigma^\leq$ as axioms. Let $\mbf{Alg}(\Sigma)$ be the class of all modal algebras that validate all sequents in $\Sigma^\vdash$. Then the sequent system $\mbf{SK}\Sigma^\vdash$, if consistent, is sound and complete with respect to $\mbf{Alg}(\Sigma)$.

For any set of modal formulas $\Sigma$, consider the set of graphical rules $\Sigma^g = \{\top\vdash \sigma(\alpha)\mid \alpha\in \Sigma\}$. Let $\mbf{K}_g\Sigma^g$ be the graphical calculus obtained from $\mbf{K}_g$ by adding all rules in $\Sigma^g$.

For $\Sigma \sub \{D, T, 4, B, 5\}$, where $D = \D\top$, $T=\Box \alpha\imp \alpha$, $4 = \Box \alpha\imp \Box\Box \alpha$, $B= \alpha\imp \Box\D \alpha$ and $5=\D\alpha\imp\Box\D\alpha$,
one can show that the calculus $\mbf{K}_g\Sigma^g$ is equivalent to $\mbf{SK}\Sigma^\vdash$ by the translation $\pi$.
The proof is similar to Theorem \ref{thm:embedding}. Moreover,
the graphical calculi $\mbf{K}_g\Sigma^g$ are sound and complete with respect to $\mbf{Alg}(\Sigma)$.

\section{Conclusion}
Graphical calculi for modal logics developed in the present paper are systematic and modular. They are modal graphical versions of Gentzen-style sequent systems. They follow closely Peirce's original presentation in another sense as well: the rules arise systematically from Peirce's presentation of broken-cut gamma graphs and their rules (R 467, 478). Only (DMN), (B) and (5) are new.\footnote{We find Peirce's own remarks suggesting that he was not keen to have (B) or (5) as rules in his modal gamma systems: ``There is not much utility in a {\em double broken cut}. Yet it may be worth notice that $\dbcut{\ g\ }$ and $g$ can neither of the be inferred from the other. The outer of the two broken cuts is not only relative to a state of information but to a state of reflection. The graph $\dbcut{\ g\ }$ asserts that it is possible that the truth of the graph $g$ is necessary. It is only because I have not sufficiently reflected upon the subject that I can have any doubt of whether it is so or not'' (R 467, 1903). The rule (5) uses a principle that is contrary to Peirce's own rules of ($T^+$) and ($T^-$). Since Peirce's preferred interpretation of the broken cut modality was an epistemic one, he would not have recommended (5) as a good rule for knowledge. The previous quotation indeed continues as follows: ``It becomes evident, in this way, that a modal proposition is a simple assertion not about the universe of things but about the universe of facts that one is in a state of information sufficient to know. The graph $\bcut{\ g\ }$ without any selective, merely asserts that there is a possible state of information in which the knower is not in a condition to know that the graph $g$ is true, while $\ncut{\ g\ }$ asserts that there is no such possible state of information''.} In the basic system $\mbf{K}_g$, identifying a vacant broken-cut with a vacant continuous cut dispenses with necessitation as a primitive rule. Moreover, the basic rules are perfectly symmetrical. Thanks to the diagrammatic syntax, graphs need not assume negation normal form. Thus there are good prospects for developing deep inference proof systems for non-normal and intuitionistic modal logics in a similar fashion. The notions of position in the areas of cuts and the polarity of positions likewise result immediately from the diagrammatic language that these systems are built upon. Thus diagrammatic syntax can be considered to be an advantage when compared to languages and notations that are used in other deep inference systems. Labels are likewise not needed.

As to some other future work, the specific sense of the cut-elimination process suggests that there are interesting decision procedures that we can get from proof searches in the proposed calculi. The desirable property is the subformula property, as well as a syntactic calculation of interpolants, among others.

\subparagraph*{Acknowledgements.}

We want to thank the three reviewers for their helpful comments. The work of the first author is supported by the Chinese National Foundation for Social Sciences and Humanities (grant no.\ 16CZX049). The work of the second author is supported by the Academy of Finland (project 1270335) and the Estonian Research Council (project PUT 1305) (Principle Investigator A.-V.\ Pietarinen).





\begin{thebibliography}{99}
\bibitem{AB1996}
G.~Allwein and J.~Barwise (eds.) {\em Logical Reasoning with Diagrams}. Oxford University Press, 1996.

\bibitem{Avron1996}
A.~Avron.
\newblock The method of hypersequents in the proof theory of propositional non-classical logics.
\newblock In: W.\ Hodges, M.\ Hyland, C.\ Steinhorn, J.\ Truss (eds.) {\em Logic: From Foundations to Applications}. Proceedings of the Logic Colloquium, Keele, UK, 1993, pp.\,1--32. Oxford University Press, New York, 1996.

\bibitem{Belnap1982}
N.~D.~Belnap.
\newblock Display logic.
\newblock {\em Journal of Philosophical Logic}, 11:375--417, 1982. \doi{10.1007/BF00284976}.

\bibitem{Berg93}
H.~van den Berg.
\newblock Modal logics for conceptual graphs.
\newblock In: {\em Proceedings of First International Conference on Conceptual Structures}. LNAI, vol. 699, pp. 413-429. Springer-Verlag, Berlin, 1993.

\bibitem{Brauner98}
T.~Bra{\"u}ner.
\newblock Peircean graphs for the modal logic {S5}.
\newblock In: M.-L.\ Mugnier and M.\ Chein (eds.), {\em Conceptural Structures: Theory, Tools and Applications, Proceedings of Sixth International Conference on Conceptual Structures}. LNAI, vol.\ 1453, pp.\ 255--269. Springer-Verlag, Berlin, 1998. \doi{10.1007/BFb0054919}

\bibitem{Brauner99}
T.~Bra{\"u}ner and P.\ {\O}hrstr{\o}m.
\newblock Towards a diagrammatic formulation of modal and temporal logic.
\newblock In: F.\ Daoud (ed.), {\em Working notes of AAAI'99 Fall Symposium on Modal and Temporal Logic-based Planning for Open Networked Multimedia Systems}. AAAI, pp.\ 61--67. AAAI Press, North Falmouth, 1999.

\bibitem{Brunnler2003}
K.~Br\"unnler.
\newblock {\em Deep inference and symmetry in classical proofs}.
\newblock PhD thesis, Technische Universit\"{a}t Dresden, 2003.

\bibitem{Brunnler2009}
K.~Br\"unnler.
\newblock Deep sequent systems for modal logic.
\newblock {\em Archive for Mathematical Logic}, 48:551--577, 2009. \doi{10.1007/s00153-009-0137-3}

\bibitem{MaPietarinen2016}
M.\ Ma and A.-V.~Pietarinen. \newblock Peirce's sequent proofs of distributivity. \newblock In: S.\ Ghosh and S.\ Prasad (eds.), \emph{Logic and Its Applications: Proceedings of the 7th Indian Logic Conference}, LNCS 10119, 2017. \doi{10.1007/978-3-662-54069-5\_13}.

\bibitem{Sara2005}
S.~Nergi.
\newblock Proof analysis in modal logic.
\newblock {\em Journal of Philosophical Logic}, 34:507--544, 2005.

\bibitem{Peirce1903a}
C.~S.~Peirce.
\newblock {\em Lowell Lectures of 1903. Lecture IV}. Manuscript at the Houghton Library of Harvard University, 1903. (R 467)

\bibitem{Peirce1903b}
C.~S.~Peirce.
\newblock {\em Lowell Lectures of 1903. Syllabus for Certain Topics of Logic}. Manuscript at the Houghton Library of Harvard University, 1903. (R 478)

\bibitem{Pietarinen2004}
A.-V.\ Pietarinen. \newblock Peirce's diagrammatic logic in {IF} perspective.
\newblock In: A.\ Blackwell, K.\ Marriott and A.\ Shimojima (eds.), {\em Diagrammatic Representation and Inference: Third International Conference, Diagrams 2004}. LNAI, vol.\ 2980, pp.\ 97--111. Springer-Verlag, Berlin, 2004. \doi{10.1007/978-3-540-25931-2\_11}

\bibitem{Pietarinen2006}
A.-V.\ Pietarinen. \newblock \emph{Signs of Logic: Peircean Themes on the Philosophy of Language, Games, and Communication}, Springer, Dordrecht, 2006.

\bibitem{Pietarinen2011}
A.-V.\ Pietarinen. \newblock Moving Pictures of Thought II: Graphs, Games, and Pragmaticism's Proof. \newblock {\em Semiotica}, 2011(186), 315--331, 2011. \doi{10.1515/semi.2011.058}

\bibitem{Pietarinen2016}
A.-V.\ Pietarinen. \newblock Extensions of Euler Diagrams in Peirce's Four Manuscripts on Logical Graphs.
\newblock In: M.\ Jamnik, Y.\ Uesaka and S.\ E.\ Schwartz (eds.), {\em Diagrammatic Representation and Inference: Ninth International Conference, Diagrams 2016}. LNAI, vol.\ 9781, pp.\ 139--156. Springer-Verlag, Berlin, 2016. \doi{10.1007/978-3-319-42333-3\_11}

\bibitem{Roberts1973}
D.~D.~Roberts.
\newblock {\em The Existential Graphs of Charles S.\ Peirce}. Mouton, The Hague, 1973. 

\bibitem{Seligman1997}
J.\ Seligman. \newblock The Logic of Correct Description. \newblock In: M.\ de Rijke (ed.), {\em Advances in Intensional Logic}, pp.\ 107-–135. Kluwer, Dordrecht, 1997. \doi{10.1007/978-94-015-8879-9\_5}

\bibitem{Strassburger2007}
L.\ Stra{\ss}burger. \newblock Deep Inference for Hybrid Logic. \newblock {\em Proceedings of International Workshop of Hybrid Logic 2007}, pp.\ 13--22.

\bibitem{Sowa1984}
J.~F.~Sowa.
\newblock {\em Conceptual Structures: Information Processing in Mind and Machine}. Addison-Wesley, Reading, 1984. \doi{10.1016/0004-3702(88)90069-0}

\bibitem{SS2005}
C.~Stewart and P.~Stouppa.
\newblock A systematic proof theory for several modal logics.
\newblock
In: R.\ Schmidt, I.\ Pratt-Hartmann, M.\ Reynolds and H.\ Wansing (eds.). {\em Advances in Modal Logic}, vol.\ 5, pp.\,309--333. King's College Publications, London, 2005.

\bibitem{Stouppa2007}
P.~Stouppa.
\newblock A deep inference system for the modal logic {S5}.
\newblock {\em Studia Logica}, 85(2):199--214, 2007. \doi{10.1007/s11225-007-9028-y}

\bibitem{Wansing2002}
H.~Wansing.
\newblock Sequent systems for modal logics.
\newblock In: Gabbay, D., Guenther, F.\ (eds.) {\em Handbook of Philosophical Logic}, vol.\,8, 2nd edition, pp.\,61--145. Kluwer, Dordrecht, 2002. \doi{10.1007/978-94-010-0387-2\_2}

\bibitem{Zeman1964}
J.~Zeman.
\newblock {\em The Graphical Logic of Charles S. Peirce}. Ph.D.\ dissertation. University of Chicago, 1964.

\bibitem{Zeman1997}
J.~Zeman.
\newblock {\em Peirce's Graphs}. In: D.\ Lukose et al.\ (eds), {\em Proceedings of Fifth International Conference on Conceptual Structures}, LNCS 1257, pp.\ 12--24. Springer-Verlag, Berlin, 1997. \doi{10.1007/BFb0027877}

\end{thebibliography}


\end{document}